\documentclass[a4paper,12pt]{article}
\usepackage[centertags]{amsmath}
\usepackage{amsmath}
\usepackage[dutch,british]{babel}
\usepackage{amsfonts}
\usepackage{amssymb}
\usepackage{amsthm}
\usepackage{newlfont}
\usepackage{color}
\usepackage{bbm}

\usepackage{stmaryrd}
\usepackage{mathrsfs}
\usepackage{pstricks,pst-node}
\usepackage{palatino}  %palatcm - a missing package
\usepackage{fancyhdr}
\usepackage{graphicx}
\usepackage{fancybox}
\usepackage{multibox}
\usepackage{geometry}

\theoremstyle{plain}
  \newtheorem{theorem}{Theorem}[section]
  \newtheorem{corollary}[theorem]{Corollary}
  
  \newtheorem{lemma}[theorem]{Lemma}
  
\theoremstyle{definition}

\theoremstyle{remark}

  \newcounter{assumptioncounter}

  \theoremstyle{changebreak}
 \newtheorem{assumption}[assumptioncounter]{Assumption}

\numberwithin{equation}{section}

\DeclareMathOperator{\Tr}{Tr}

\newcommand\otimesal{\mathop{\hbox{\raise 1.6 ex
  \hbox{$\scriptscriptstyle\mathrm{al}$}
\kern -0.92 em \hbox{$\otimes$}}}}
\newcommand\oplusal{\mathop{\hbox{\raise 1.6 ex
  \hbox{$\scriptscriptstyle\mathrm{al}$}
\kern -0.92 em \hbox{$\oplus$}}}}
\newcommand\Gammal{\hbox{\raise 1.7 ex
\hbox{$\scriptscriptstyle\mathrm{al}$}\kern -0.50 em $\Gamma$}}
\renewcommand\i{\mathrm{i}}

\let\al=\alpha \let\be=\beta \let\de=\delta \let\ep=\epsilon

  \let\ga=\gamma 
\let\ka=\kappa \let\la=\lambda \let\om=\omega 
\let\si=\sigma

 \let\Ga=\Gamma  \let\Om=\Omega

\newcommand{\caB}{{\mathcal B}}

\newcommand{\caD}{{\mathcal D}}
\newcommand{\caE}{{\mathcal E}}
\newcommand{\caF}{{\mathcal F}}
\newcommand{\caG}{{\mathcal G}}
\newcommand{\caH}{{\mathcal H}}
\newcommand{\caI}{{\mathcal I}}
\newcommand{\caJ}{{\mathcal J}}
\newcommand{\caK}{{\mathcal K}}

\newcommand{\caT}{{\mathcal T}}

\newcommand{\caV}{{\mathcal V}}
\newcommand{\caW}{{\mathcal W}}
\newcommand{\caX}{{\mathcal X}}
\newcommand{\caY}{{\mathcal Y}}
\newcommand{\caZ}{{\mathcal Z}}

\newcommand{\bbC}{{\mathbb C}}

\newcommand{\bbE}{{\mathbb E}}

\newcommand{\bbN}{{\mathbb N}}

\newcommand{\bbP}{{\mathbb P}}

\newcommand{\bbR}{{\mathbb R}}

\newcommand{\opunit}{\text{1}\kern-0.22em\text{l}}

\newcommand{\frh}{{\mathfrak h}}

\newcommand{\frA}{{\mathfrak A}}

\newcommand{\frU}{{\mathfrak U}}

\newcommand{\e}{{\mathrm e}}
\newcommand{\iu}{{\mathrm i}}
\renewcommand{\d}{{\mathrm d}}

\newcommand{\res}{{\mathrm R}}

\renewcommand{\sp}{\mathrm{sp}}

\newcommand{\beq}{ \begin{equation} }
\newcommand{\eeq}{ \end{equation} }
\newcommand{\bet}{ \begin{theorem} }
\newcommand{\eet}{ \end{theorem} }

\newcommand{\Symm}{\mathrm{Sym}}

\newcommand{\dsi}{\si}
\newcommand{\dsig}{\Omega}

\newcommand{\basi}{\tilde{\si}}
\newcommand{\basig}{\tilde{\Omega}}

 \newcounter{smallarabics}
\newenvironment{arabicenumerate}
{\begin{list}{{\normalfont\textrm{\arabic{smallarabics})}}}
  {\usecounter{smallarabics}\setlength{\itemindent}{0cm}
  \setlength{\leftmargin}{5ex}\setlength{\labelwidth}{4ex}
  \setlength{\topsep}{0.75\parsep}\setlength{\partopsep}{0ex}
   \setlength{\itemsep}{0ex}}}
{\end{list}}

\newcounter{smallroman}

\newcommand{\ben}{\begin{arabicenumerate}}
\newcommand{\een}{\end{arabicenumerate}}

\newcommand{\Pair}{\mathrm{Pair}}

\newcommand{\sfock}{\Ga_{\mathrm{s}}}
\newcommand{\fock}{\Ga}

\newcommand{\norm}{ \|}
\newcommand{\str}{ |}
\newcommand{\lakl}{\lambda^{-2} }

\newcommand{\embed}{ \caI_{\uparrow} }
\newcommand{\compress}{  \caI_{\downarrow }}

\begin{document}

\begin{center}
\large{ \bf{ Large deviation generating function for energy transport in the Pauli-Fierz model}} \\
\vspace{15pt} \normalsize

{\bf Wojciech De Roeck}\footnote{Postdoctoral Fellow FWO-Flanders,
email: {\tt
 wojciech.deroeck@fys.kuleuven.be}}\\
Insitute for Theoretical Physics, K.U.Leuven, Belgium\\
Insitute for Theoretical Physics, ETH Zurich, Switzerland
\\

\end{center}

\vspace{20pt} \footnotesize \noindent {\bf Abstract: } We consider
a finite quantum system coupled to quasifree thermal reservoirs at
different temperatures. Under the assumptions of small coupling
and exponential decay of the reservoir correlation function, the
large deviation generating function of energy transport into the reservoirs is shown to be analytic on a
bounded set. Our method is different from the spectral deformation
technique which was employed recently in the study of
spin-boson-like models. As a corollary, we derive the
Gallavotti-Cohen fluctuation relation for the entropy production and a central limit theorem for energy transport.

\vspace{5pt} \footnotesize \noindent {\bf KEY WORDS:} , Gallavotti-Cohen symmetry, nonequilibrium statistical
mechanics, spin-boson model \vspace{20pt}
\normalsize
\section{Introduction}

\subsection{Fluctuations in open quantum systems}

Recently, the physics community has shown quite some interest in current fluctuations in nonequilibrium quantum systems.
We mention two interesting perspectives:
\ben
\item{ Since the work of \cite{evanscohen, gallavotticohen95prl}, it has become clear that nonequilibrium systems, both classical and quantum, exhibit a symmetry in the fluctuations of entropy production.
This symmetry, dubbed the  ``Gallavotti-Cohen Fluctuation Theorem" holds  far for equilibrium.}
\item{It has been realized \cite{lesovikexcess} that noise between electron contacts shows distinct signs of Fermi statistics,
studies of this kind go by the name of "Full counting statistics".}
\een

Perhaps the most important promise of fluctuation theory is in the construction of nonequilibrium statistical mechanics: Via the study of the large deviation rate function, one hopes to find a useful variational principle describing nonequilibrium stationary states. Recent papers taking part in this project are e.g.\ \cite{bertinidesolelatticegases,maesnetocnyminent,deriddareview}.

In this paper, we study heat current fluctuations in a nonequilibrium model of the type  'spin-boson'.  
We prove that the large deviation generating function corresponding to  energy transport exists in a bounded (but arbitrarily large) set around $0$ and that it is analytic.  

\subsection{Large deviation generating function}

We briefly sketch the framework of large deviations.

Assume that we have a family of  $\bbR^d$-valued random
variables $A_{t}$, indexed by time $t \in \bbR^+$ and with distribution given by the expectation
$\bbE_{ t  }$.  To fix thoughts, one can think of the $A_t$  as time-integrals of some variable $ a(s),  0\leq s \leq t$, i.e.\ 
\beq
A_t :=  \int_0^t \d s \,  a(s)
\eeq
The large deviation generating function on $\bbR^d$ (if it exists) is defined as \beq
\label{def: intro generating function} f(\ka):= \lim_{t \uparrow
\infty}\frac{1}{t}\log \bbE_{ t  }( \e^{-(\ka \str
A_t)}) \eeq where $(\cdot \str \cdot)$ is the canonical scalar
product on $\bbR^d$. 

From the function $f(\ka)$, one can extract large time properties of the observables $A_t$, as sketched below, see
\cite{dembozeitouni} for precise statements and details.
\ben
\item{ If $f$ is analytic in a neighbourhood of $0$, then $A_t$ satisfies a central limit theorem with mean $- \frac{\partial f}{\partial \ka } (0)$ and covariance $ \si=  \frac{\partial^2 f}{\partial \ka^2 }(0)$ (the gradient and the Hessian of $f$). 
Let \beq b_t:= \frac{1}{\sqrt{t}} \Big( A_t+ t \frac{\partial f}{\partial \ka } (0) \Big) ,\eeq then
\beq 
 \lim_{t \uparrow \infty} \bbE_{t} (\e^{- \i (\ga \str b_t)    }    )=\e^{-(\ga  \str  \si  \ga) },  \qquad \ga \in \bbR^d
 \eeq
% where the covariance matrix $\si$ is given be the Hessian of $f$ in $\ka=0$; 
% \beq
% \si=  \frac{\partial^2 f}{\partial \ka^2 }(0)
% \eeq
 }
 \item{ If $f$ is differentiable on $\bbR^d$, then the family $\frac{A_t}{t}$ satisfies a \emph{large deviation principle} with rate function $I(\al)$ given by 
 \beq
 I(\al):= - \inf_{\ka \in \bbR^d}  ( (\ka \str \al) + f(\ka)) 
 \eeq
  Heuristically, this means that 
 \beq \mathrm{Prob}_t(\frac{A_t}{t}\approx \alpha) \sim   \e^{-t I(\alpha)},\qquad \alpha \in \bbR^d, t \uparrow
\infty \eeq
(in a logarithmic sense) as $t \uparrow \infty$.  
    }
\een

In classical statistical mechanics, the existence  of the large deviation generating function can usually be established through a convexity argument, see e.g.\  \cite{simon}. A similar general understanding is lacking in quantum
statistical mechanics (see however \cite{aboesalem, hiaimosonytomohiro, reybelletld, netocnyredig} for partial results). Another -even conceptual- problem in quantum statistical mechanics, is how to describe joint large deviations of several noncommuting variables. Remark that it was exactly to solve such a conceptual problem for the central limit theorem, that the framework of the \emph{fluctuation algebra} was constructed \cite{goderisverbeurevets}.

We consider a quantum setup where $A_t$  corresponds to the total  heat transport into reservoirs. Hence the setup is somewhat different from that in \cite{aboesalem, hiaimosonytomohiro, reybelletld, netocnyredig}; the expectation  $\bbE_{}(g(A_t))$ for some function $g$ can not be formulated as an expectation of some observable in a quantum state, rather it is the probability of obtaining certain (differences of) measurement outcomes. The problem of joint distributions for non-commuting observables does not even appear in this context since the different reservoir Hamiltonians do mutually commute.
This is discussed  more extensively in \cite{derezinskideroeckmaes}.
Our result will establish the existence and analyticity of  $f(\ka)$ on a compact (but arbitrarily large) set containing $0$. Hence, we do not prove the large deviation principle (but we do prove the central limit theorem). The Gallavotti-Cohen fluctuation theorem is a simple corollary of our result.

\subsection{Open quantum systems with finite reservoirs}\label{sec: finite}

Our model describes a small quantum system (an atom) interacting with a quantum system with many degrees of freedom (a reservoir). We choose the reservoir as simple as possible: a free field of bosons, although fermions would do just as well\footnote{In fact, they would simplify the technical work}. The system is coupled to the reservoirs through a term, which is linear in the field creation and annihilation operators. This type of models are known as Pauli-Fierz models, or, in the simplest case, the spin-boson model. These models arise as toy-models in solid state physics, were the bosons are lattice phonons, or through the dipole approximation in QED, where the bosons are photons, see \cite{derezinski1} for more background.

 To make the statements mathematically sharp, we consider this field in the thermodynamic limit, or equivalently, in the limit where the modes form a continuum.  However, for the sake of distilling the right physical question addressed in this paper, we start from a finite-volume setup.

\subsubsection{Setup}\label{sec: finite setup}
Fix a finite-dimensional Hilbert space $\caE$ with
self-adjoint Hamiltonian $E$ and let $\caK$ be a finite set
which indexes the heat reservoirs at inverse temperatures
$\be_{k\in \caK}
>0$.  The superscript $n \in
\bbN$ indicates that the thermodynamic limit ($n \nearrow \infty$) has not yet been
taken. See also Section \ref{sec: notation} for specific notation and conventions.
 To each $k \in \caK$, we associate 
 
\ben
\item{A finite-dimensional one-particle Hilbert space $\frh_{k,n}$ and its bosonical second quantization $\sfock (\frh_{k,n})$. }
\item{The coupling operator $V_{k,n} \in \caB(\caE, \caE \otimes \frh_{k,n})$. }

\item{A self-adjoint one-particle Hamiltonian $h_{k,n}$ acting on $\frh_{k,n}$ with corresponding second quantization $\d \Ga (h_{k,n}) $.
}
\item{A Gibbs state $ \rho_{k,\be_k,n}$ on $\caB(\sfock (\frh_{k,n}))$ at inverse temperature $\be_k$
\beq \rho_{k,\be_k,n}\left[ R  \right] =
\frac{\Tr \left[ \e^{-\be_k  \d \Ga (h_{k,n})} R
\right]}{\Tr \left[ \e^{-\be_k  \d \Ga (h_{k,n})} 
\right]}, \qquad    R \in  \caB(\sfock (\frh_{k,n}))  \eeq }
\een
We define the total interacting Hamiltonian on $\caE
\otimes_{k \in \caK} \sfock (\frh_{k,n}) $ as \begin{eqnarray}
H_{\la,n} &=& E + \sum_{k \in \caK} \d \Ga (h_{k,n})+ \la\sum_{k
\in \caK}   \left( a^*( V_{k,n}) +  a( V_{k,n}) \right). 
\end{eqnarray} 
 We take as
initial state \beq \label{def: initial state finite} \rho_{\caE}
\otimes \rho_{\res_n }, \qquad \rho_{\res_n } := \mathop{\otimes}\limits_{k\in \caK}
 \rho_{k,\be_k,n} \eeq 
corresponding to initially decorrelated reservoirs and an arbitary state $\rho_\caE$ on $\caB(\caE)$.

\subsubsection{Transport fluctuations and their limits}  \label{sec: transport}
 We introduced the finite volume systems in order to pick the right expression for transport fluctuations, and hence, now that all tools are in place, we ask what we mean by transport fluctuations in the finite-volume models. 
 
Note that the reservoir Hamiltonians $ \d \Ga (h_{k,n})$  mutually
commute and that they have discrete spectrum. Hence one can measure them simultaneously in the beginning and at the end of an experiment. To determine the transport (of energy), we look at the differences of those measurement values. 
Let $\caT := \prod_{k \in \caK} \sp ( \d \Ga (h_{k,n})) $ and let $P_{x\in \caT}$ be the joint spectral projections of $\d \Ga (h_{k,n})$ corresponding to the eigenvalues 
$x=(x_k)_{k \in \caK}$.
The standard interpretation of quantum mechanics yields the probabilities
\beq \label{def: prob}
\bbP_{\rho_{\caE},t, \la,n} (y ) :=   \sum_{ x,x' \in \caT, x'-x=y} \rho_{\caE} \otimes \rho_{\res_n} \big[  P_x    \e^{-\i t H_{n,\la }}     P_{x'}        \e^{\i t H_{n,\la }}  P_x   \big]
\eeq
for observing energy differences $y \in \bbR^{|\caK|}$.
The Fourier-Laplace transform of this measure has a nice expression which is better suited for taking the thermodynamic limit:  Using that (the density matrix corresponding to) $\rho_{\res_n}$ commutes with the spectral projections $P_x$, one arrives at 
\beq \label{eq: F finite}
 \int_{\bbR^{|\caK|}} \d y \, \bbP_{\rho_{\caE},t, \la,n} ( y ) \e^{-(\ka \str y)} 
= \rho_{\caE} \otimes \rho_{\res_n} \big[  \Ga (w_{-\ka,n}) e^{\i t H_{\la,n} }   \Ga (w_{\ka,n}) \e^{-\i t H_{\la,n} }   \big]
\eeq
where \beq w_{\ka,n}= (\mathop{\oplus}\limits_{k \in \caK} \e^{-\ka_k h_{k,n}}) \eeq
We will study the infinite-volume limit of this expression, given by \eqref{def: functions analytic cont}  and introduced in Section \ref{sec: positive temp}. In Section \ref{sec: discussion} we will substantiate the claim that  \eqref{def: functions analytic cont}  is indeed the $n \uparrow \infty$-limit of   \eqref{eq: F finite}.

This approach to fluctuations was already used in \cite{kurchanquantum, deroeckmaes, jarzynskiquantum, talknerlutzhanggi}
 for fluctuations of heat and work, and, most widespread, in
\cite{lesovikexcess,levitovjmp} for fluctuations of charge
transport (``Full counting statistics''), made mathematically transparant in \cite{klich,avronbachmann}.

\subsection{Conventions and Notation}\label{sec: notation}

For $\caE$ a Hilbert space, we use the standard notation for $1\leq p <\infty $
\beq \caB_p(\caE):= \{  S \in \caB(\caE), \Tr\left[(S^*S)^{p/2}\right] < \infty  \}\eeq
and
\beq
\norm S \norm_p := (\Tr\left[(S^*S)^{p/2}\right])^{1/p}
\eeq

For a Hilbert space $\frh$ we write  
\beq
\sfock^n(\frh):= \Symm_n \otimes^n \frh , \qquad \sfock(\frh):= \mathop{\oplus}\limits_{n \in \bbN} \sfock^n(\frh)
\eeq
where $\Symm_n$ projects on the fully symmetrized subspace and $\sfock(\frh)$ is the bosonic Fock space built on $\frh$.  For operators $C$ on $\frh$, we write (whenever the RHS is well-defined as an operator on $\sfock(\frh)$)
\beq
  \fock(C)=  \mathop{\bigoplus}\limits_{n \in \bbN} \otimes^n C   
\eeq
\beq
C \mapsto  \d \fock(C)= \mathop{\bigoplus}\limits_{n \in \bbN} \sum_{i =1}^n  1\otimes \ldots  1 \otimes  \mathop{C} \limits_{\textrm{$i$'th position} } \otimes 1  \ldots \otimes 1
\eeq
For $W \in \caB(\caE,\caE \otimes \frh)$, we use
 the generalized creation and annihilation operators 
$a(W)/a(W^*)$  on $\caE \otimes \sfock(\frh)$,  (see \cite{derezinski1} for  an extensive review of this notation).
If, for some $\psi \in \frh$ and $D \in \caB(\caE)$,
\beq
W u = D u \otimes \psi, \qquad u \in \caE
\eeq
then $a^*(W)= D\otimes a^*(\psi)$ where $a^*(\psi)$ is the more familiar creation operator 
. 

For a Hilbert space ${\frh}$, we write  $\overline{\frh}$ for its conjugate space, which is fixed by an antiunitary map ${\frh} \to \overline{\frh}:  a \mapsto \bar{a}$. 
If $\frh=L^2(\caX,\bbC)$ for some measure space $\caX$,  the map $a \mapsto \bar{a}$ is identified with the complex conjugation on functions $\caX \to \bbC$.
If $R \in \caB(\frh)$, then $\overline{R} \in \caB(\overline{\frh})$ is defined by $\overline{R}\overline{a}=\overline{Ra}$.

For $\ka \in \bbC^d$, we write 
\beq
\Re \ka = (\Re \ka_1,\ldots, \Re \ka_d), \qquad  \Im \ka = (\Im \ka_1,\ldots, \Im \ka_d)
\eeq

For indicator functions, we use the notation $\mathrm{Ind}(\cdot)$, i.e.\ for a premise $\al(x)$ dependent on some variable $x$
\beq
\mathrm{Ind}(\al(x)) = \left\{ \begin{array}{ccl} 1   &\textrm{if}  & \al(x) \, \textrm{is true}   \\  0     &\textrm{if} &   \al(x) \, \textrm{is false} \end{array} \right.
\eeq

\subsection{Outline}

We introduce the model in abstract terms in Sections \ref{sec: zerotemp} and \ref{sec: positive temp}, immediately followed by the result in Section \ref{sec: result}. 
The physical justification of this model is given in Section \ref{sec: thermodynamic limit}, where it is explained how it emerges from the quantities discussed in Sections \ref{sec: finite setup} and \ref{sec: transport}.
In Section \ref{sec: comparison}, we discuss related results in the literature.  The rest of the paper is devoted to the proofs. Section \ref{sec: proof of main theorem} contains the main line of reasoning and the technical lemma's are postponed to Section \ref{sec: estimates}. 
The main idea is in Lemma \ref{lem: spectral deformation} whose main ingredient is Lemma \ref{lem: bound caW}.

\section{Model and results}\label{sec: model}

% \section{Models of open quantum systems}\label{sec: open quantum systems}
\subsection{Zero-temperature objects} \label{sec: zerotemp}
Introduce a finite-dimensional space $\caE$ with a self-adjoint Hamiltonian $E$ and  (for each $k \in \caK$) one-particle spaces ${\frh}_{k,\infty} $ with a self-adjoint operator $h_{k,\infty} $  on  ${\frh}_{k,\infty} $. We also need coupling operators ${V}_{k,\infty} \in \caB(\caE,\caE\otimes \frh_{k,\infty})$.

One should think of these objects as defining the zero-temperature Hamiltonian of the subsystem+reservoir system, formally
\beq
H_{\la,\infty}:= E +  \mathop{\sum}\limits_{k \in \caK} \d \Ga (h_{k,\infty} )  +  \la \mathop{\sum}\limits_{k \in \caK} \left(a(V_{k,\infty} )+a^*(V_{k,\infty} )\right)
\eeq
The heavy notation with the subscript $\infty$ is because in what follows, more natural infinite-volume objects are introduced. The  objects with subscript $\infty$ are relevant at $\be=\infty$.

We  anticipate the finite-temperature by introducing the Bose-density operators
\beq 
{\zeta}_{k,\infty} :=  (\e^{\be_k {h}_{k,\infty}}-1  )^{-1}
 \eeq 

We will use the above notation to build a dynamical system which represents our system at positive temperature. The connection between the finite-volume objects, introduced in Section \ref{sec: finite}, and the inifinite-volume model, is given in Section \ref{sec: discussion} 

\subsection{Positive temperatures}\label{sec: positive temp}
Define
\beq 
\left.
\begin{array}{rclrcl}
 \frh_k &:=& {\frh}_{k,\infty} \oplus \overline{{\frh}_{k,\infty} }  &    \frh&:=& \mathop{\oplus}\limits_{k \in \caK} \frh_k
  \\ 
h_k&:=&   {h}_{k,\infty}  \oplus      ( -{\overline{h_{k,\infty} }}) &    h&:=&   \mathop{\oplus}\limits_{k \in \caK} h_k \\
V_k&:=&   \sqrt{1+  {\zeta}_{k,\infty} }{V}_{k,\infty} \oplus \overline{ \sqrt{ {\zeta}_{k,\infty} }{V}_{k,\infty} } &  V&:=&  \mathop{\oplus}\limits_{k \in \caK} V_k  \\
w_{k,\ka_k} &:=&  \e^{-\ka_k h }  &    w_\ka &:=&  \mathop{\oplus}\limits_{k \in \caK} w_{k,\ka_k} 
 \end{array}
 \right.
 \eeq

Let the total Hilbert space be $\caH:=\caE \otimes \sfock(\frh)$ and define on $\caE \otimes \big(\caD(\d \Ga (h)) \cap\caD( a(V)+a^*(V) )  \big)$, \beq  \label{def: ham}
H_\la:= E + \d \Ga (h)+ \la \left( a( V   )+ a^*(V) \right)
\eeq
The following theorem comes from \cite{derezinskideroeck2}
\bet \label{thm: def ham}
Assume that $\norm V \norm <\infty$. Let $H_0:= E+ \d \Ga (h)$ and denote 
\beq
I(u):=\e^{\i u H_0} (a(V)+a^*(V))      \e^{- \i u H_0}.
\eeq
The series
\beq
U_t^\la := \e^{\i t H_0}  \sum_{n\in \bbN} \mathop{\int}\limits_{0\leq u_1\ldots  u_n \leq t}  \d u_1\ldots  \d u_n \,   I(u_n) \ldots  I(u_1),
\eeq
originally defined on the dense subspace $\caD_1$ (see Section \ref{sec: construction dynamics}),
extends to a strongly continuous unitary group on $\caH$, which will also be denoted $U_t^\la $.
Its self-adjoint generator  is an extension of $H_\la$ as in \eqref{def: ham} and will be simply called $H_\la$ in what follows.
\eet

Let $\rho_{\caE}$ be a state on $\caB(\caE)$ and let $\rho_\res$ be the state on $\caB(\sfock(\frh))$ given by 
\beq \rho_\res[\cdot ]=   \langle \Omega ,\, \cdot\,  \Omega \rangle  \eeq where $\Omega=1\oplus 0\oplus0\ldots \in \sfock(\frh)$ is the vacuum vector.  We will take 
$
\rho_\caE \otimes \rho_\res$ as initial state on $\caB(\caH)$ for our dynamics. Unless otherwise stated, we assume $\rho_\caE$ to be arbitrary.

We now introduce our main object of study
 \begin{assumption}[Bounded interaction]\label{ass: laplace} For all $\ka$ with $\Re \ka  \in D $,
 \beq   \norm w_{\frac{\ka}{2}} V \norm < \infty, \qquad   \norm w_{-\frac{\overline{\ka}}{2}} V \norm < \infty \eeq
    \end{assumption}
    The following lemma follows from Section \ref{sec: construction dynamics}.
    \begin{lemma}\label{lem: definition W}
 Assume there is an open set $D \subset \bbR^{|\caK|}$ with $0 \in D$ such that Assumption \eqref{ass: laplace} is satisfied and let $U_t^\la$ be as in Theorem \ref{thm: def ham}.  Than the function \beq \label{def: functions analytic cont}    \ka \mapsto
 \rho_{\caE} \otimes \rho_\res \left[
\Ga (w_{-\ka}) U_{-t}^{\la} \Ga (w_{\ka}) U_{t}^{\la} \right] 
 \eeq  has an analytical continuation from
 $\{ \Re \ka =0\}$ into $   \{ \Re \ka  \in D\} $.
\end{lemma}

The function \eqref{def: functions analytic cont} should be thought of as the Fourier-Laplace transform of the probability distribution of energy transport. This is discussed and justified in Section \ref{sec: discussion}.

  \subsection{Results} \label{sec: result}

To continue, we need additional assumptions.  The next assumption basically establishes that the operator $h$ on $\frh$ has absolutely continuous spectrum.
\begin{assumption}\label{ass: reservoir ham} 
There are measure spaces $(\caX,\d x)$ and $(\caY,\d y)$ such that $\frh= L^2(\caX,\d x)$,  $\caX=\caY\times \bbR$ and $\d x= \d y \d \xi$ where $\d \xi$ is the Lesbegue measure on $\bbR$. For $  (y,\xi) =x \in \caX$, we write $\xi(x)=\xi $ for the projection on $\bbR$. 
 The operator $h$ acts by multiplication with $ \xi(x)$,
\beq (h \psi)(x)= \xi(x)\psi(x),\qquad \psi \in \frh \eeq
\end{assumption}
%Note that sometimes, one calls  $ L^2(\caY,\d y)$ the space of internal degrees of freedom, which can be  slightly misleading since e.g.\ a quantity like the direction of the momentum  is also included in those degrees of freedom. 
%
Remark that one can associate to $V$  a measurable function $\caX  \to \caB(\caE) $, which we denote $x \mapsto V(x)$ and which satisfies
\beq   \langle v \otimes \psi,  V u \rangle_{\caE \otimes \frh} =  \int_{\caX} \d x  \,   \overline{\psi(x)}  \langle v, V(x) u \rangle_{\caE}  , \qquad u,v \in \caE, \psi \in \frh  \label{def: function V} \eeq
Define the reservoir time-correlation function
 \beq \label{def: correlation} p_\ka(t) := \sup_{S \in \caB(\caE), \| S\|=1 } \|V^* S \e^{-\i t h }   w_{\ka}V \|   \eeq
  \begin{assumption}[Decay of bath correlations] \label{ass: correlation} 
  Let $p_{\ka}$ be as defined above. There are $C, \al > 0$ such that for $\ka \in D$
  \beq 
   p_\ka (t) \leq  C \e^{-\al \str t \str }, \qquad  t \in \bbR
\eeq
    \end{assumption}
    
    Introduce the set of Bohr frequencies  $ \caF:= \sp E -\sp E  $ and let $1_{\caE_{e}}$ stand for the spectral projection of $E$ on $e \in \sp E$.
   
The following assumption expresses that the coupling between system and reservoir is sufficiently effective. 

\begin{assumption}[Fermi Golden Rule]\label{ass: fermi golden rule} 
\mbox{}
\begin{enumerate}
\item For all $\om \in \caF$ and $\d y$-almost all  $y \in \caY$, the function $V: \caX \mapsto \caB(\caE)$ is continuous  on the set $\{ x=(y, \xi) \str \, \xi =\om \}$. This implies  that 
$V(x=(y,\om))$ is well-defined.
 \item If  $S \in \caB(\caE)$ satisfies
 \beq
 \sum_{\om \in \caF } \sum_{\footnotesize{\left.\begin{array}{cc} e,e' \in
\sp E, \\ \om=e-e'\end{array}\right.} }  \int_{\caY} \d y  \left\norm \big[ S,  1_{\caE_{e}}  V(y,\om) 1_{\caE_{e'}} \big] \right\norm =0 \eeq
then $S=c1$ for some $c \in \bbC$.
\end{enumerate}
\end{assumption}

Now comes our main theorem

\bet\label{thm: main}
Assume Assumptions \ref{ass: laplace}, \ref{ass: reservoir ham}, \ref{ass: correlation} and  \ref{ass: fermi golden rule}. There is a $\la_0>0$ such that for $\la \in [-\la_0,\la_{0}] $ and $\ka \in D$,
 \beq
f(\ka,\la):= \lim_{t \uparrow \infty}  t^{-1} \log  \rho_{\caE} \otimes \rho_\res \left[
\Ga (w_{-\ka}) U_{-t}^{\la} \Ga (w_{\ka}) U_{t}^{\la} \right] 
 \eeq
exists, is independent of $\rho_\caE$  and  real-analytic in $\ka $ and $\la$.
\eet

As the main corrolary, we state the Gallavotti-Cohen fluctuation theorem for the entropy production. This requires an additional assumption
\begin{assumption}[Time-reversal invariance]\label{ass: time-reversal}
There is a anti-unitary $\Theta $ on $\caH$ such that  for all $\la \in \bbR$,
\beq
 \Theta^{-1} H_{\la}  \Theta= H_{\la} ,\qquad    \Theta^{-1}  \d \fock (h_{k})   \Theta=  \d \fock (h_{k}) , \qquad  \Theta (\caE\otimes\Om)= (\caE\otimes\Om)
\eeq
and $\Theta$ is an involution, i.e.\ $\Theta^{-1}=\Theta$.
\end{assumption}

\bet\label{thm: gc}
Assume \ref{ass: laplace}, \ref{ass: reservoir ham},  \ref{ass: correlation}, \ref{ass: fermi golden rule} and \ref{ass: time-reversal}. Let $f(\ka,\la)$ be as in Theorem \ref{thm: main} and define for $\nu \in \bbR$
\beq
\ka(\nu)=  \nu  (\be_1, \be_2 ,\ldots, \be_{\str \caK \str} )   \in \bbR^{\str \caK \str}
\eeq
For $\nu \in \bbR$ such that $\ka(\nu), \ka(1-\nu) \in D$,
 \beq\label{eq: gc}
f(\ka(\nu),\la)= f(\ka(1-\nu),\la)
 \eeq
\eet

%It is the symmetry \eqref{eq: gc} which is usually referred to as the fluctuation symmetry for the entropy production.

By Bochner's theorem, there is a nonnegative Borel measure
$\d \bbP_{\rho_\caE,t,\la}$ on $\bbR^{|\caK|}$ such that
 \beq \rho_\caE \otimes
\rho_\res  \left[
\Ga (w_{-\ka}) U_{-t}^{\la} \Ga (w_{\ka}) U_{t}^{\la} \right] 
 =
\int_{\bbR^{|\caK|}} \d \bbP_{\rho_\caE,t,\la}(y)  \e^{-(\ka \str y) }
 \eeq
for $ {\Re \ka \in D}$. Putting $\ka=0$, one sees that
$\d \bbP_{\rho_\caE,t,\la}$ is a probability measure.  It is the infinite-volume analogue of the probabilities $\bbP_{\rho_\caE,t,\la,n}$ introduced in Section \ref{sec: transport}. 
We  write $\bbE_{\rho_\caE,t,\la}[\cdot]$ for the expectation w.r.t.\  $\d \bbP_{\rho_\caE,t,\la}$. 

The $\bbR^d$-valued random variable $y=(y_k)$ is interpreted as the energy transport into the distinct reservoirs.
Remark that in thermodynamics, one interpretes $S:=  \sum_{k \in \caK} \be_k y_k $ as the  entropy production.

Since 
\beq
f(\ka(\nu),\la) =\lim_{t \uparrow \infty} \frac{1}{t} \log \bbE_{\rho_\caE,t,\la}[ \e^{-\nu  S }]
\eeq
one sees that $f(\ka(\nu),\la) $  is indeed related to (large) fluctuations of the entropy production.

The following corollary follows from Theorem \ref{thm: main} by \cite{bryc}. 
\begin{corollary}
Assume the assumptions of Theorem \ref{thm: main}. Then the $\bbR^d$-valued random variable $y$ satisfies a central limit theorem with mean $-\frac{\partial}{ \partial \ka} f(\ka,\la)\big\str_{\ka=0} $ and covariance $ \si_\la:=\frac{\partial^2}{ \partial \ka^2} f(\ka,\la)\big\str_{\ka=0}$. Let
\beq
b_t :=    \frac{1}{\sqrt{t}} \left(y + t \frac{\partial}{ \partial \ka} f(\ka,\la)\big\str_{\ka=0}  \right),
\eeq
then
\beq
\bbE_{\rho_\caE,t,\la} [ \e^{-\i (\ga \str b_t)}  ] \mathop{ \longrightarrow}\limits_{t \uparrow \infty}  \e^{-(\ga \str \si_\la \ga)}, \qquad  \ga \in \bbR^{\str\caK\str}
\eeq
\end{corollary}
%\begin{corollary}
%Assume the assumptions of Theorem \ref{thm: main}. Then the $\bbR^d$-valued random variables
%\beq
%b_t :=    \frac{1}{\sqrt{t}} \left(y + t \frac{\partial}{ \partial \ka} f(\ka,\la)\big\str_{\ka=0}  \right) 
%\eeq
%satisfy the CLT with covariance matrix 
%$
%\si_\la:=   \frac{\partial^2}{ \partial \ka^2} f(\ka,\la)\big\str_{\ka=0}
%$,  that is
%\beq
%\bbE_{\rho_\caE,t,\la} [ \e^{-\i (\ga \str b_t)}  ] \mathop{ \longrightarrow}\limits_{t \uparrow \infty}  \e^{-(\ga \str \si_\la \ga)}, \qquad  \ga \in \bbR^{\str\caK\str}
%\eeq

%\end{corollary}

The expectation value $-\frac{\partial}{ \partial \ka} f(\ka,\la)\big\str_{\ka=0} $ and the covariance $\frac{\partial^2}{ \partial \ka^2} f(\ka,\la)\big\str_{\ka=0}$ can be written in a more familiar form.
Introduce the operators
\beq
\triangle_{k,t} := U_{-t}^\la  \d\fock(h_k) U_{t}^\la -  \d\fock(h_k) ,     
\eeq
Then
\begin{eqnarray}
-\frac{\partial}{ \partial \ka} f(\ka,\la)\big\str_{\ka=0}  &=&  \lim_{t \uparrow \infty} \frac{1}{t}  \rho_\caE \otimes
\rho_\res [ \triangle_{k,t} ]   =:     \langle \triangle_{k} \rangle  \\
\frac{\partial^2}{ \partial \ka_k \partial_{\ka_{k'}}} f(\ka,\la)\big\str_{\ka=0}  &=&  \lim_{t \uparrow \infty} \frac{1}{t}  \rho_\caE \otimes
\rho_\res \left[ \left(\triangle_{k,t}- t \langle \triangle_{k} \rangle \right)     \left(\triangle_{k',t}-  t\langle \triangle_{k'} \rangle \right)  \right]  
\end{eqnarray}
where the convergence of the expressions on the RHS is a consequence of the analyticity of $f(\ka,\la)$.
However, it is not true in general (beyond second order in $\ka$) that
\beq
 \rho_\caE \otimes
\rho_\res  \left[
\Ga (w_{-\ka}) U_{-t}^{\la} \Ga (w_{\ka}) U_{t}^{\la} \right] 
=   \rho_\caE \otimes
\rho_\res [\e^{- \sum_k \ka_k  \triangle_{k,t}}].
\eeq
See \cite{derezinskideroeckmaes} for a thorough discussion of different approaches to quantum fluctuations.

\section{Discussion} \label{sec: discussion}

\subsection{Initial state}\label{sec: initial state}
We formulate our result only for particular intitial states, namely $\rho_\caE \otimes \rho_\res$ with $\rho_\res$ the vacuum state.  One could ask whether Theorem \ref{thm: main} still holds for a different inital state. 
In fact, by a slight generalization of our method, one can prove (see e.g.\ the previous version of the present paper) that the same result holds if one replaces 
\beq \label{rem: object study} \rho_{\caE} \otimes \rho_\res \left[
\Ga (w_{-\ka}) U_{-t}^{\la} \Ga (w_{\ka}) U_{t}^{\la} \right]  \eeq
by 
\beq \label{rem: object study s} \rho_{\caE} \otimes \rho_\res \left[  U_{-s}^{\la}
\Ga (w_{-\ka}) U_{-t}^{\la} \Ga (w_{\ka}) U_{t}^{\la}   U_{s}^{\la}  \right]  \eeq
for arbitrary $s$.  That is,  $f(\ka,\la)$ is  independent of $s$.

However, in  Section \ref{sec: transport}, one sees that the very choice of our object of study \eqref{rem: object study} depends on the fact that $\rho_\res$ is 'diagonal' in the operators $\d\fock(h_k)$. This (or rather, its finite-volume analogue) is used in going from \eqref{def: prob} to \eqref{eq: F finite}. Expressed more dramatically, an expression like \eqref{rem: object study s} does not appear!

\subsection{Thermodynamic limit} \label{sec: thermodynamic limit}
We skipped over a thorough justification of the object \eqref{def: functions analytic cont}, which features in our results. We remedy this by telling in which sense the dynamical system is the infinite-volume version of the finite-volume systems and how the expression \eqref{def: functions analytic cont} emerges.
Usually,  thermodynamical limits are constructed by specifying volumes which go to infinity in some sense (e.g.\ in the sense of Van Hove). In our case, such an explicit setup is not necessary (though of course possible).  
 We simply demand  the following relation between the finite-volume objects and the objects introduced in Section \ref{sec: zerotemp}.
 \begin{assumption}[Thermodynamic limit of finite-volume models]\label{ass: thermo limit}
Let
  \beq g_{1,t}(x)= \frac{\e^{-\i t x} }{\e^{\be_k x}-1}  ,  \qquad  g_{2,t}(x)=  \e^{-\i t x}( {1-\frac{1}{\e^{\be_k x}-1}}), \qquad t \in \bbR, x \in \bbR^+
  \eeq
For $S \in \caB(\caE)$, $i=1,2$, we have
\beq 
\norm {V}^*_{k,\infty}   S  g_{i,t}( h_{k,\infty} )  {V}_{k,\infty} \norm < \infty
\eeq
and
\beq \label{eq: thermo limit}
{V}^*_{k,n} S  g_{i,t}( h_{k,n}) V_{k,n}   \mathop{\longrightarrow}\limits_{n \uparrow \infty}  {V}^*_{k,\infty}   S  g_{i,t}( h_{k,\infty} )  {V}_{k,\infty} 
\eeq 
uniformly on compacts in $t \in \bbR$
\end{assumption}

If Assumption \ref{ass: thermo limit} is satisfied, a large class of correlation functions converges. There is quite some  arbitrariness in this statement, which is usually not considered in the literature.

 Define 
\begin{eqnarray}
\Phi_{k,n} (t) &:=&   \e^{\i t H_{0,n}  } (a(V_{k,n} )+a^{*}(V_{k,n} ))  \e^{-\i t H_{0,n}  } \\
\Phi_k(t) &;=&   \e^{\i t H_{0 } } (a(V_{k} )+a^{*}(V_{k} ))  \e^{-\i t H_{0 } } \\
\end{eqnarray}
Assume Assumption \ref{ass: thermo limit}, then for all $t,t' \in \bbR$  and $S \in \caB(\caE)$, 
\beq \label{eq: thermo limit2}
 \rho_{\caE} \otimes \om_{\res,n} \big[    \Phi_{k,n} (t) S   \Phi_{k',n}(t')  \big] \quad  \mathop{\longrightarrow}\limits_{n \uparrow \infty}  \quad  \rho_{\caE} \otimes \rho_\res \big[    \Phi_{k}(t)  S \Phi_{k'}(t')   \big]  
\eeq
Of course, from  \eqref{eq: thermo limit2} one deduces also convergence of higher-order correlation functions (since the states $\om_{\res,n}$ and $\rho_\res$ are quasifree, those are expressed in terms of the second order correlation function). In particular, one has also convergence of the same correlation functions with the time-dependence now given by the fully interacting evolution, that is, let 
\begin{eqnarray}
\Phi^{I}_{k,n} (t) &:=&   \e^{\i t H_{\la,n } }  \e^{-\i t H_{0,n}  } \Phi_{k,n} (t)   \e^{\i t H_{0,n}  } \e^{-\i t H_{\la,n } } \\
\Phi^{I}_k(t) &;=&   \e^{\i t H_{\la } }   \e^{-\i t H_{0 } } \Phi_k(t)  \e^{\i t H_{0 } } \e^{-\i t H_{\la} } ,
\end{eqnarray}
then equation \eqref{eq: thermo limit2} holds with $\Phi^{I}$ replacing $\Phi$, as follows from a Dyson expansion, e.g.\  \eqref{def: series kappa}.

 It is now straightforward to see that Assumptions \ref{ass: laplace} and \ref{ass: thermo limit}  imply 
 \beq\label{conv of generating functions}
\rho_{\caE} \otimes \rho_{\res_n} \big[  \Ga (w_{-\ka,n}) e^{\i t H_{\la,n} }   \Ga (w_{\ka,n}) \e^{-\i t H_{\la,n} }   \big]  \quad  \mathop{\longrightarrow}\limits_{n \uparrow \infty} \quad   \rho_{\caE} \otimes \rho_\res \left[
\Ga (w_{-\ka}) U_{-t}^{\la} \Ga (w_{\ka}) U_{t}^{\la} \right]  
\eeq
where the LHS was introduced through physical considerations in Section \ref{sec: transport}.

The critical reader might wonder why there is in our presentation no mention of $W^*$-algebra's, which often play a promiment role in the mathematical formulation of statistical mechanics. 
If one defines the Araki-Woods algebra $\frA$ as in Section \ref{sec: proof of gc}, one finds that the dynamics 
\beq
\frA \ni A \mapsto  U_{-s}^{\la}
A U_{s}^{\la}, \qquad s \in \bbR
\eeq
leaves $\frA$ invariant.  Physically, one should  restrict the state $\rho_\caE \otimes \rho_\res$, originally defined on $\caB(\caE \otimes \sfock(\frh))$,  to $\frA$. 
However, in our approach, it is neither mathematically nor physically necessary to consider this restriction. We study  the expression \eqref{rem: object study}, which is well-defined and whose motivation is via 
\eqref{conv of generating functions}.

For the same reasons, we do not have to ask ourselves whether the operator \eqref{def: ham} is the right choice. In the literature, this operator is called the \emph{semi-standard} Liouvillean, but one can also consider the \emph{standard} Liouvillean.  Again, the resolution of any possible ambiguity is via finite-volume limits. That being said, it might be worth remarking that \eqref{rem: object study} can be expressed as the expectation of powers of a relative modular operator, see \cite{matsuitasaki}, thus providing a more algebraic starting point for our work. Another possible approach is in \cite{avronbachmann}, where the expression \eqref{rem: object study} is constructed (for fermions) via different, but essentially equivalent reasoning.

% Since we deliberately chose to justify our setup from finite-volume limits, we do not need to fit expression \eqref{def: functions analytic cont} into an algebraic formalism. Nor do we need to worry whether $H_\la$ is the right choice of generator of the dynamics. In the literature, several distinct generators have been proposed, our choice would be called the 'semi-Liouvillian' since we essentially used a GNS-representation for the baths, but not for the small system. 

\subsection{Comparison with other works}\label{sec: comparison}

There has lately been a lot of work on spin-boson and spin-fermion models, or more general, Pauli-Fierz models.

We feel our work is technically closest to   \cite{jaksicpillet2}, in which one considers the spin-boson model and one proves that the generator of the dynamics has absolutely continuous spectrum for $\la \neq 0$, except for one eigenvalue which corresponds to the stationary state. The other eigenvalues of the system at $\la=0$ turn into resonances whose location is in first nonvanishing order predicted by the Lindblad generator. 
The assumptions are very similar; to allow for a comparison, we assume that 
\beq \label{connectionVpsi}
V u= D u \otimes  \psi 
\eeq
for some $D \in \caB(\caE)$ and $\psi \in \frh \sim L^2(\bbR,L^2(\caY,\d y))$, in which case $a^*(V)=D \otimes a^*(\psi)$.

The basic assumption in \cite{jaksicpillet2} reads 
\begin{assumption}[analytic coupling]\label{ass: analytic coupling}
 The function $\psi$ is analytic in a strip $\{\Im z \leq \delta \}$ and   \beq \sup_{\ga \in [-\delta,\delta] } \int_\bbR \d \xi  \norm \psi(\xi+\i \ga )\norm^2 < \infty  \eeq
\end{assumption}
Assumption \ref{ass: analytic coupling} implies that  \beq \left\str \int_\bbR \d \xi \norm \psi(\xi)\norm^2  \e^{-\i t \xi } \right \str  \leq C \e^{- t \delta}  \eeq 
 which is just Assumption \ref{ass: correlation} for $\ka=0$.
However, the $\ka$ have no analogue in \cite{jaksicpillet2} and we would need to assume Assumption \ref{ass: analytic coupling}  with $\psi$ derived through  \eqref{connectionVpsi} from $ w_{\frac{\ka}{2}}V$ rather than from  $V$.

In contrast, we do not need any additional  infrared condition on $\xi \mapsto \psi(\xi)$, contrary to \cite{jaksicpillet2}. This is because we construct the dynamics via the Dyson expansion instead of via the Nelson commutator theorem. 
Physically speaking\footnote{That is, in terms of the zero-temperature coupling operator, or 'form-factor' $V_{k,\infty} $}, there is of course already an infrared condition present since
\beq
  \norm  V \norm < \infty      \quad   \Rightarrow \quad   \sum_{k \in \caK} \norm ( \be_k h_{k,\infty} )^{-1/2}V_{k,\infty}  \norm < \infty
\eeq
with the notation as in Section \ref{sec: zerotemp}.
 
 The technique of \cite{jaksicpillet2} consists of a spectral deformation of the generator $H_\la$. We employ time-dependent perturbation theory and we rewrite the Dyson expansion as a one-dimensional polymer model. This is embodied in Lemma \ref{lem: bound caW}. Starting from that lemma, one can obtain our result through a simple cluster expansion (as in the previous version of this paper).  However, since the polymer model is one-dimensional, we can apply  the transfer-matrix technique. In dealing with the transfer matrix, we use a variant of the spectral deformation technique, such that our technique is not as different from \cite{jaksicpillet2} as might seem.
 
Assumption \ref{ass: correlation} cannot be weakened without changing the method drastically. Note that one cannot assume Assumption \ref{ass: correlation}  for $D=\bbR^{\str \caK \str}$ since that would imply that 
\beq  \bbR^{\str \caK \str} \ni \ka  \mapsto p_{\ka}(t) \eeq is a bounded analytic function, hence constant. 
%
%The same is true for the fact that Assumption \ref{ass: correlation} is assumed uniformly in $\ka$, which means that it cannot be satisfied for all of $\bbC^{\str \caK \str}$.
% In principle, one can strenghten Assumption \ref{ass: fermi golden rule} such as to assure that the gap of $L_\la$ grows exponentially in $\str \ka \str $, allowing more room for the perturbation theory.
%However, the  parameters which have to be controlled uniformly in $\ka$ include  $\norm p_\ka \norm_\infty$ and even for a Gaussian coupling function $\xi \mapsto \psi(\xi)$ (notation as above), this grows like $\ka \mapsto \e^{\str\ka \str^2}$.

Results that need weaker regularity properties of $\psi(\xi)$ are e.g.\ \cite{bachfrohlichreturn}, \cite{derezinskijaksicreturn,derezinskijaksicspectral}, 
\cite{froehlichmerkli}. In those works one employs Mourre theory or renormalization group techniques, however they do not permit to localize the resonances.

A different type of works are those using scattering theory. This approach was initiated in \cite{ruelleness}, but so far, it has not been successful for spin-boson type models, although it works well for junctions \cite{froehlichmerkliueltschi}.

From the physical point of view, our result is closer to \cite{jaksicpillet3,jaksicogatapilletspinfermion} where one studies a non-equilibrium setup and one derives approach to a non-equilibrium steady state and the Green-Kubo relations, or to \cite{jaksicpautratpillet}, where one studies a form of the central limit theorem. See \cite{derezinskideroeckmaes} for an extensive discussion of the difference and similarities of different approaches to quantum fluctuations and central limits.

\section{Proof of Theorems \ref{thm: main} and  \ref{thm: gc}} \label{sec: proof of main theorem}

\subsection{Construction of the dynamics}\label{sec: construction dynamics}

Let $1_{n}$ be the projector on  $\caE \otimes \sfock^n(\frh)$ (the $n$-particle sector, see Section \ref{sec: notation}) and let the domain $\caD_1 \subset \caE \otimes \sfock (\frh)$ be defined by 
 \beq  \psi \in \caD_1 \Leftrightarrow  \exists C >0 :   \norm  1_{n} (\psi) \norm \leq \frac{C^n}{\sqrt{n!}} \eeq
Let $H_0:= E+ \d \Ga (h)$ and 
\beq
I_{\ka}(u):=\e^{\i u H_0} (a(w_{-\overline{\ka}}V )+a^*( w_{\ka}V))      \e^{- \i u H_0}.
\eeq

For  $\Re \ka =0$, the series
\beq \label{def: series kappa}
\e^{\i t H_0}  \sum_{n\in \bbN} \mathop{\int}\limits_{0\leq u_1\ldots  u_n \leq t}  \d u_1\ldots  \d u_n   I_{\frac{\ka}{2}}(u_n) \ldots  I_{\frac{\ka}{2}}(u_1),
\eeq
originally defined on $\caD_1$, extends  to the unitary group (Theorem 6.1 \cite{derezinskideroeck2})
  \beq \label{def: dyson kappa}
\Ga (w_{\frac{\ka}{2}}) U_{t}^{\la} \Ga (w_{-\frac{\ka}{2}}) 
  \eeq
Since the argument in  \cite{derezinskideroeck2} showing that \eqref{def: series kappa} is a a strongly continuous group on $\caD_1$, depends only on the assumption $\norm V \norm <\infty$, this remains true for $\ka$ satisfying Assumption  \ref{ass: laplace}, and \eqref{def: series kappa} can be taken as the definition of \eqref{def: dyson kappa}.\\

 In what follows and unless stated otherwise, we will assume  that Assumptions \ref{ass: laplace} and \ref{ass: reservoir ham}  are satisfied and
that $\Re \ka \in D$.

\subsection{Dynamics and notation on $\caB_1(\caH)$}  \label{sec: dynamics and notation on trace-class}
It is advantageous to rewrite the object of study in a slightly more abstract way.
Let   ${\caD}_{1,\otimes}$ stand for the subspace of $\caB_1(\caE \otimes \sfock(\frh))$ defined by 
finite linear combinations of $\str \phi_1 \rangle \langle \phi_2 \str$ for $ \phi_1,\phi_2 \in \caD_1 $.
From the conclusions of Section \ref{sec: construction dynamics}, it follows that
\beq \label{def: Z}
A \mapsto \left( \Ga (w_{\frac{\ka}{2}}) U_{t}^{\la} \Ga (w_{-\frac{\ka}{2}}) \right) \, A \,  \left(\Ga (w_{-\frac{\ka}{2}}) U_{-t}^{\la} \Ga (w_{\frac{\ka}{2}}) \right)    =: Z^{{\ka,\la}}_t(A)
\eeq
maps ${\caD}_{1,\otimes}$ into itself.  
In what follows, we write 
\beq
M (S):= \i [E,S] 
\eeq
as a bounded operator on $\caB(\caE)$.   
We define the embedding $ \embed: \caB(\caE) \to \caB_1(\caE \otimes \sfock(\frh))  $
\beq
  S  \mapsto  \embed (S) = S \otimes \str \Om\rangle \langle  \Om \str  
\eeq
and the compression $\compress: \caB_1(\caE \otimes \sfock(\frh))  \to \caB(\caE)$
\beq 
(S \otimes R)  \mapsto  \compress  (S \otimes R) = S \Tr_{{\sfock(\frh)}}[R].
\eeq
with $ \Tr_{{\sfock(\frh)}} $ the trace on $\caB_1(\Tr_{{\sfock(\frh)}})$ (Hence $\compress$ is actually a partial trace).

We have hence rewritten 
\beq \label{eq: rewritten object}
 \rho_{\caE} \otimes \rho_\res \left[
\Ga (w_{-\ka}) U_{-t}^{\la} \Ga (w_{\ka}) U_{t}^{\la} \right] = \Tr\left[ \left( \compress    Z^{{\ka,\la}}_t    \embed  \right)(\tilde{\rho}_\caE) \right]
\eeq
where  $ \tilde{\rho}_\caE$ is the density matrix, corresponding to the state $\rho_\caE$, i.e.\   $\rho_\caE \left[  S \right] =  \Tr \left[ \tilde{\rho}_\caE  S  \right] $.

\subsection{The deformed Lindblad generator}

If $\norm p_{\ka} \norm_{1}:= \int \d t \, p_{\ka}(t) <\infty$ for $\ka=0$,  we can define
\beq\label{def: generator}
 \Upsilon =
-\iu\sum_{\omega\in\caF}\ \sum_{e-e'=\omega}\int_0^\infty
 1_{\caE_e} V^*  1_{\caE_{e'}}\e^{-\iu t (h
-\omega)} V  \, 1_{\caE_e} \d t.
 \eeq
Assuming additionally the first statement of Assumption \ref{ass: fermi golden rule}, we introduce the deformed Lindblad generator. For $S \in
\caB(\caE)$, let
\beq L_\ka(S)= -\iu ( \Upsilon S - S  \Upsilon^*) +2 \pi  \sum_{\omega\in\caF}\ \sum_{e-e'=\omega}
 1_{\caE_e} V^*  1_{\caE_{e'}} \,  S \,  \delta (h-\om)  \,      w_\ka  V  \, 1_{\caE_e}   \label{def: lind}\eeq
 where the operator-valued Dirac-delta distribution $\delta(\cdot)$ is well-defined by the continuity assumption in Assumption \ref{ass: fermi golden rule}. For example, one can take a sequence of functions converging in the sense of distributions to $\delta(\cdot-\om)$, then the mentioned continuity assumption assures convergence in \eqref{def: lind}
One checks, see e.g.\ \cite{derezinskideroeck2}, that for $\ka=0$, or equivalently, $w_\ka=1$, we recover the usual definition for the Lindblad generator, which satisfies
\beq
  \Tr [ L_{\ka=0}(S) ]=0
\eeq
However, since the second term in \eqref{def: lind} is a completely positive map, it follows that $\e^{t L_{\ka}}$ is a completely positive semigroup for $\{\Re \ka \in D\}$.

We need the following properties of $L_{\ka}$.
\bet \label{thm:  properties L}  
Let $L_{\ka}$ be as in \eqref{def: lind} and $M$ as defined in  Section \ref{sec: dynamics and notation on trace-class}.  
\ben
\item{ Assume Assumption \ref{ass: fermi golden rule} and fix a $\tau <0$. The operator $\e^{\tau L_\ka}$ has a maximal simple eigenvalue $\e^{\tau f_{\ka}} $ with $f(\ka) \in \bbR$ and there is a 'gap'  $g_\ka >0$ such that
\beq \label{eq: gap}
 \sup \{    \str z \str ,   z \in \sp (\e^{\tau L_\ka}) \setminus \e^{\tau f_{\ka}}    \} < \e^{\tau f_{\ka}} (1 -\e^{-\tau g_\ka})
\eeq
 The eigenvector corresponding to $\e^{\tau f_\ka}$ can be chosen a positive invertible operator.
 } 
\item{ \beq  [L_\ka, M]=0 \eeq } 
\item{ Assume $\norm p_{\ka} \norm_{1} <\infty$. For all $\tau >0$,
\beq   \norm  \compress Z_{\lakl \tau}^{\ka,\la} \embed - \e^{-\i \tau(  \lakl M+\i L_\ka ) } \norm \mathop{\longrightarrow}_{\la \downarrow 0} 0    \eeq
 where the LHS is continuous in $\la,\ka,\tau$. } 
\een 
\eet
Statement (1) of Theorem \ref{thm:  properties L} 
is the only place where we use the second statement of  Assumption \ref{ass: fermi golden rule}. It is a non-degeneracy assumption which enters the non-commutative Perron-Frobenius theorem.

\subsection{Dyson  expansion and transfer operator}

Our basic tool is a rearranged Dyson expansion, whose properties are collected in the upcoming Lemma \ref{lem: bound caW}.
Fix a parameter $\tau>0$ and define 
on $\caB(\caE)$ for $n \in \bbN_0$,
\beq \label{def: caW}
\caW_n^{\ka,\la,\tau} := \compress \, Z_{\lakl \tau }^{\ka,\la} \,\mathop{(1-  \embed \compress)}\limits_{{n-1}} \ldots \mathop{(1-  \embed \compress)}\limits_{2}\, Z_{\lakl \tau }^{\ka,\la} \,\mathop{(1-  \embed \compress)}\limits_{1}\, Z_{\lakl \tau }^{\ka,\la} \,\embed
\eeq
($n-1$ factors of $(1-  \embed \compress)$ inserted).  The definition \eqref{def: caW}  makes sense since $Z_{t }^{\ka,\la}$ maps $\caD_{1,\otimes}$ into itself (see Section \ref{sec: construction dynamics}) and, obviously, $\embed \caB(\caE) \subset \caD_{1,\otimes}$. Whenever reasonable, we will  abbreviate $\caW_n= \caW_n^{\ka,\la,\tau}$.
\begin{lemma} \label{lem: bound caW}  Let $\caW_n= \caW_n^{\ka,\la,\tau}$ be as above.
\mbox{}\ben
\item{ For all $m \in \bbN_{0}$,
\begin{eqnarray}
\compress Z^{\ka,\la}_{\lakl m \tau} \embed=  \sum_{r \in \bbN}
\sum_{
\left.
\begin{array}{c}
\sum_{i=1}^r n_i =m
\end{array}
\right.
}         \caW_{n_r }    \ldots  \caW_{n_2 }  \caW_{n_1 }
\end{eqnarray}
}

\item{Assume Assumption \ref{ass: correlation}. There is $c:=c(\ka,\la,\tau)>0$, vanishing as $\la \downarrow 0$ and continuous in the three parameters,  such that for $n>1$, 
\beq \norm \caW_n \norm \leq c^{n-1}  \eeq
}
\een

\end{lemma}

In what follows, we use the Hilbert space $l_2(\bbN_0) \otimes \caB_2(\caE)$.
  Let for $n \in \bbN_0$,  $e_n$ be the canonical $n$'th base vector in $l_{2}(\bbN_0)$ and let S be the unilateral shift, defined by  (setting $e_{0}:=0$)
\beq S e_n = e_{n-1} 
\eeq
Recall that $\caE$ is finite-dimensional, which allows to define
 the embedding $P_n :  \caB(\caE) \to l_2(\bbN_0) \otimes \caB_2(\caE): \,  u \mapsto e_{n} \otimes u$ and compression $P^*_{n}:  e_{n} \otimes u \mapsto u$.
We are led to examine the following operator on $l_{2}(\bbN_0) \otimes \caB_2(\caE)$; 
\beq
T^{\ka,\la,\tau}=     \sum_{n\in \bbN_0}   P_n  \caW_n P^*_1  + S\otimes 1
\eeq
From Lemma \ref{lem: bound caW}(1), one has
\beq \label{relation ZT}
\compress Z_{\ka,\la}^{\lakl m \tau}  \embed= P^*_1 (T^{\ka,\la,\tau})^m   P_1 
\eeq

If the operator $T:=T^{\ka,\la,\tau} $ had a maximal eigenvalue, isolated from the rest of the spectrum, we could easily estimate the $n\nearrow \infty$ asymptotics of \eqref{relation ZT}. However, upon realizing that 
\beq \sp S = \{z \in \bbC, \str z \str \leq 1\}\eeq
this surely fails at $\ka=0$, since the highest eigenvalue of $\e^{\tau L_{\ka=0}}$ is $1$. This difficulty is addressed in the next section.

\subsection{Spectral deformation}

Introduce the unbounded operator  \beq R =\sum_{n\in \bbN_0}  n \, P_n P^*_n . \eeq 
The following statements are straightforward.
\begin{lemma} \label{lem: properties R}
For $\delta \in \bbR$ and $W \in \caB(\caB(\caE)) $,
\ben 
\item{  $\e^{\delta R}  S   \e^{-\delta R}= \e^{-\delta}  S$ } 
 \item{   $     \e^{\delta R }   P_n W  P^*_m \e^{-\delta R}     =  \e^{(n-m) \delta} P_n W  P^*_m  $    }
 \item{  $  P^*_1 T^m P_1   =  P^*_1 \left( \e^{\delta R} T\e^{-\delta R}  \right)^m P_1 $  }
\een
\end{lemma}
Most importantly, the operator $\e^{\delta R} T \e^{-\delta R}$ does have an isolated eigenvalue for well-chosen $\delta$, as we show now.
\begin{lemma} \label{lem: spectral deformation}
Let $\hat{\delta}:= -1/2 \ln c(\ka,\la,\tau)$, where the latter was introduced in Lemma \ref{lem: bound caW}. There is a $\la_0>0$  such that for $ \la \in [-\la_{0},\la_{0}]   $,  
$\ka \in D$ and $\tau$ varying in some compact set $D_{\tau}$,
 the operator \beq \e^{ \hat{\delta}R} \,  T \, \e^{-\hat{\delta} R}  \label{def: deformed operator}\eeq  
% and 
% \beq 
% \sup \{    \str z \str ,   z \in \sp \left(\e^{-\delta(\la,\tau) R} T^{\ka,\la,\tau} \e^{-\delta(\la,\tau) R}\right) \setminus \e^{\tau f_{\ka,\la,\tau}}    \} < f_{\ka,\la,\tau} -g_{\ka,\la,\tau}
%\eeq
% 
 has a maximal simple eigenvalue $\e^{\tau f_{\ka,\la,\tau}}$ with $ f_{\ka,\la,\tau} \in \bbR$. There is $g_{\ka,\la,\tau}>0$ such that
 \beq \label{eq: gap2}
 \sup \{    \str z \str ,   z \in \sp(   e^{\hat{\delta} R}  T \e^{-\hat{\delta}R}  )  \setminus \e^{\tau f_{\ka,\la,\tau}}   \} < \e^{\tau f_{\ka,\la,\tau}} (1 -\e^{-\tau g_{\ka,\la,\tau}}).
 \eeq
 The eigenvector $G_{\ka,\la,\tau} $ corresponding to this eigenvalue can be chosen such that  $P^*_1 G_{\ka,\la,\tau} \in \caB(\caE)$  is an invertible, positive operator. The function $f_{\ka,\la,\tau}$ is real-analytic in $\ka \in D$, $\str \la \str \leq \la_{0}$ and $\tau \in D_{\tau}$.
\end{lemma}
\begin{proof}
By Lemma \ref{lem: properties R}, 
\beq
\e^{\delta R} T \e^{-\delta R} = \e^{-\i \tau (\lakl M +\i L_\ka)  }  + \triangle T  \eeq
where 
\beq
\triangle T :=\e^{-\delta} S   +   (\caW_1- \e^{-\i \tau (\lakl M +\i L_\ka)  } ) +\sum_{n>1} e^{(n-1) \delta}   P_n  \caW_n P^*_1
\eeq
By Lemma \ref{lem: bound caW} (let  $c$ be as defined therein) and assuming $\str c \e^{\delta}\str <1$,
 \beq \label{eq: smallness perturbation}
 \norm \triangle T \norm \leq \e^{-\delta}  + \norm \caW_1- \e^{-\i \tau (\lakl M +\i L_\ka)  } \norm +  (\norm P^*_1 \norm \sup_{n \in \bbN_0} \norm P_n \norm )\frac{c\e^{\delta } } {1- c \e^{\delta } }  \eeq
The norms $\norm P^*_n \norm, \norm P_n \norm$ are independent of $n$ and finite since $\dim \caE$ is finite,
and hence, using Theorem \ref{thm:  properties L}(3),  $\norm \triangle T \norm$ vanishes as $ \la \downarrow 0$ and as $\de=\hat{\delta}$.\\

 Remark that $M$ is self-adjoint on $\caB_2(\caE)$ and that $\sp M=\caF$. By Theorem \ref{thm:  properties L}(2), we can hence decompose $L_\ka= \oplus_{\om \in \caF} L_{\ka,\om}$ where  $L_{\ka,\om}$ acts on the $\om$-eigenspace of $M$.  Hence
 \beq \label{eq: decomposed resolvent L}
\big(z-  \e^{-\i \tau ( \lakl M +\i L_\ka )  }\big)^{-1}=  \mathop{\oplus}\limits_{\om \in \caF} \e^{\i \tau \lakl \om}\big(\e^{\i  \tau \lakl  \om }z-  \e^{\tau L_{\ka,\om}} \big)^{-1}
\eeq
Theorem \ref{thm:  properties L}(1), the expression \eqref{eq: decomposed resolvent L} and compactness of the unit circle in $\bbC$ yield  that there is a $\ep>0, C>0$,  such that for  $ \e^{\tau f_{\ka}}-\ep  < \str z \str
 < \e^{\tau f_{\ka}}$,  for $\ka \in D$ and for $\tau$ varying over some compact set,
\beq
 \norm \big(z-  \e^{-\i \tau ( \lakl M +\i L_\ka )  } \big)^{-1}  \norm   \leq C (\str z \str -\e^{\tau f_{\ka}})^{-1}
 \eeq
The existence of an isolated eigenvalue and positivity of the eigenvector now follows from \eqref{eq: smallness perturbation} by standard perturbation theory, see e.g. \cite{katoperturbation}.  Positivity of the eigenvalue follows since by \eqref{def: Z}, $Z^{{\ka,\la}}_t$ is a completely positive map for $\Im \ka =0$.

 Real Analyticity in $\ka$ and $\la$ for $\la \neq 0$ follows from analyticity of $L_{\ka}$ and $\triangle T$, both of which are straightforward consequences of Assumption \ref{ass: laplace}.
 Since $\e^{\i t \lakl M}$ doesnot have a limit as $\la \downarrow 0$, analyticity at $\la=0$ is not immediate.
 However, since, $f_{\ka,\la,\tau}$ is analytic for  $\la \neq 0$ and continuous at $\la=0$, it is analytic. 
 \end{proof}

By Lemma \ref{lem: spectral deformation}, we get for $m$ large enough
 \beq\label{eq: discretized thm}
 \frac{1}{\tau m} \log \compress Z_{\ka,\la}^{\lakl m \tau} \embed   =    f_{\ka,\la,\tau} +  \frac{1}{\tau m}\log{ \left ( P^*_{1} P_{G_{\ka,\la,\tau}} P_{1}  + O(\e^{-m \tau  g_{\ka,\la,\tau}}) \right)}
 \eeq
 where $P_{G_{\ka,\la,\tau}}$ is the projection on $G_{\ka,\la,\tau}$. 
 
Taking $\tau,\tau' \in D_\tau$ such that  $m\tau=m' \tau'$ for some $m,m' \in \bbN$,  we get from \eqref{eq: discretized thm} that $f_{\ka,\la,\tau}=f_{\ka,\la,\tau'}$. Since $f_{\ka,\la,\tau}$ is also continuous in $\tau$, it is constant and we write  $f_{\ka,\la}:= f_{\ka,\la,\tau}$.  Theorem \ref{thm: main} now follows with $f(\ka,\la)= \la^2 f_{\ka,\la}$ by \eqref{eq: rewritten object} .

\subsection{Proof of Theorem \ref{thm: gc}} \label{sec: proof of gc}
Assume that 
\beq
\rho_{\caE} [S]=  \frac{1}{\dim \caE} \Tr[S] 
\eeq
Let $\frU$ be the $W^*$-algebra (Von Neumann- algebra) which is generated by the sets
\beq
\caB(\caE) \otimes 1\quad  \textrm{and} \quad     \{ \e^{\i (a(\psi)+a^*(\psi))},\psi \in  \frh  \}
\eeq
Remark that the expansion \eqref{def: series kappa} shows that for all $t$,
\beq
\e^{\i t H_0} \e^{-\i t H_\la} \in \frU.
\eeq
(See \cite{derezinski1} for details on $W^*$-algebra's).  Extend the notation $\ka(\nu)$ in Theorem \ref{thm: gc} to $\nu \in \bbC$. The maps of automorphisms
\beq
\frU \ni A \mapsto \eta_{s} (A):= \Ga(w_{\ka(\i s)})  A \Ga(w_{\ka(-\i s)}) , \qquad s \in \bbR
\eeq
is a $W^*$-dynamics and $\rho_{\caE} \otimes \rho_\res$ is a $1$-KMS state wrt.\ to this dynamics.  This can be easily checked or read in the literature, see again  \cite{derezinski1}.
Then, the KMS-condition reads that for $A,B \in \frU$, the function  
\beq 
\rho_{\caE} \otimes \rho_\res [ A  \eta_{s} (B)   ] 
\eeq
is analytic in $\{0 \leq \Im s \leq 1 \}$ and satisfies
\beq
 \rho_{\caE} \otimes \rho_\res [   \eta_{s} (A) B   ] =   \rho_{\caE} \otimes \rho_\res [  B   \eta_{s+\i }(A)     ] 
\eeq
Choosing $A=  \e^{-\i t H_0} \e^{\i t H_\la} $ and $B= \e^{-\i t H_\la} \e^{\i t H_0} $, inserting $1=\Theta \Theta$, using Assumption \ref{ass: time-reversal}, the general property
$\rho ( C^*) =\overline{\rho (C)}$ (true for every state $\rho$),  $[\e^{-\i t H_0}, \Ga(w_\ka)]=0$ and the invariance of $\rho_{\caE} \otimes \rho_\res$  under the dynamics  $\e^{-\i t H_0} \cdot \e^{\i t H_0} $, one gets the relation 
\beq
 \rho_{\caE} \otimes \rho_\res  \left[             \eta_{-\i \nu }(U_{-t}) U_{t}  \right]  =    \rho_{\caE} \otimes \rho_\res  \left[      \eta_{-\i 
(1-\nu) }(U_{-t}) U_{t}               \right]
\eeq
for $-1 \leq \nu \leq 0$. This is extended by analyticity to values of $\nu$ such that $\ka(\nu) \in D$. Theorem \ref{thm: gc} follows since by Theorem \ref{thm: main}, $f(\ka,\la)$ is indepenent of $\rho_\caE$.

\section{Proof of some estimates}\label{sec: estimates}

We prove the lemma's that were used in Section \ref{sec: proof of main theorem}. As in Section  \ref{sec: proof of main theorem}, we always assume Assumptions  \ref{ass: laplace} and \ref{ass: reservoir ham} and we take $\ka$ such that $\{  \Re \ka \in D\}$ where $D$ is as in Assumption  \ref{ass: laplace} .

\subsection{The Wick-representation of the dynamics on $\caB(\caE)$} \label{sec: wick representation}

The aim of this section is to introduce a convenient notation to handle the Wick-ordered Dyson expansion, stated in (\ref{def: dysonpaired1}-\ref{def: dysonpaired2}) . The result is equation \eqref{eq: dysondsi}.

Recall the representation of $V$ as a function $V: \caX \to \caB(\caE)$, introduced in \ref{def: function V}.
Denote
 \beq  
  V^{\#}_t(x):= \e^{\i t E} V^{\#}(x) \e^{-\i t E }
 \qquad   t \in \bbR , x \in \caX , \quad  V^{\#}(x)= V(x),(V(x))^*
\eeq
% \beq   \left.
%\begin{array}{ccl}
% V_t(x)&:=& \e^{\i t E} V(x) \e^{-\i t E } \\  V^*_t(x)&:= & \e^{\i t E} V^*(x) \e^{-\i t E }
% \end{array}
%\right.  \qquad   t \in \bbR , x \in \caX 
%\eeq
By Assumption \ref{ass: reservoir ham}, both $h$ and $w_\ka$ can be represented as multiplication operators with functions on $\caX$. We will denote these functions by respectively 
$\xi(x)$ and $w_\ka(x)$ (consistent with the use of $\xi$ in Assumption  \ref{ass: reservoir ham}).

 Introduce the space $\caZ=\caX \times \{1,2,3,4\}$ with elements $z=(x,j)$ and measure $\d z = \d x \d j$ ($\d j$ stands for the counting measure on \{1,2,3,4\})  and the maps $Q^{\ka}_{u \in \bbR, z \in \caZ} \in \caB(\caB(\caE)) $,
\beq \label{def: explicitQ}
Q^{\ka}_{u,  z=(x,j) }( S) = \left\{
\begin{array}{ rcclcl}
  \e^{ -\i u \xi(x) }   w_{\frac{\ka}{2}}(x) &V_u(x) &\, S \, &           & \qquad & j=1    \\
  \e^{ \i u \xi(x) }  w_{-\frac{\overline{{\ka}}}{2}}(x)& V^{*}_{u} (x)& \, S\, &    & \qquad &   j=2  \\
   \e^{- \i u \xi(x) }  w_{-\frac{{{\ka}}}{2}}(x)  & &\, S \,&     V_{u}(x)  & \qquad &   j=3    \\
   \e^{\i u  \xi(x) }      w_{\frac{\overline{{\ka}}}{2}}(x) & &\, S \,&   V^*_{u}(x)   & \qquad &     j=4 
\end{array}
\right.
\eeq

 We now introduce the pairing coefficient $C(z,z')$ for $z,z' \in \caZ$; \beq \label{def: pairing coeff}
 C(z=x,j ; z'=x',j'):=  \de(x-x')  \left \{
\begin{array}{cc}
    1   & j=1, \, \left\{ \begin{array}{c}   j'= 2  \\  j'= 4    \end{array} \right.  \,\,\textrm{ or }\, \, j=4, \, \left\{ \begin{array}{c}   j'= 1 \\  j'= 3    \end{array} \right.   \\
   0   &     \textrm{otherwise}    
\end{array}
\right.
 \eeq 
 For $n \in 2 \bbN$, let $\Pair(n)$ denote the set of partitions of
$\{1,\dots,n\}$ in pairs. For such a partition $\pi \in \Pair(n)$, we write \beq (i,i') \rightarrow \pi \,   \Leftrightarrow \, \left\{ \begin{array}{l} (i,i') \textrm{ is one of the pairs in the partition  } \pi \\  i'>i \end{array} \right. \eeq
The following representation for $ \compress Z_t^{\ka,\la} \embed$ is our starting point.
%
%\beq 
%\left.
%\begin{array}{rcc}
%   \compress Z_{\lakl t} ^{\ka,\la} \embed & =&   \e^{\i \lakl M } \,  \mathop{\sum}\limits_{n \in  2\bbN}      \mathop{\int}\limits_{0 \leq u_1 \leq  \ldots  \leq u_n \leq t } \d u_1 \ldots  \d u_n     \mathop{\sum}\limits_{\pi \in \Pair(2n)}            \label{def: pairings}  \\
%& &    \la^{-n} \mathop{\prod}\limits_{(i,i') \in \pi} C(z_i,z_{i'})    \mathop{\int}\limits_{\caZ^n } \d z_1 \ldots \d z_n        \,         Q_{\lakl u_n,z_n}  \ldots   Q_{\lakl u_1,z_1}   
%\end{array}
%\right.
%\eeq
%
\begin{eqnarray}
   \compress Z_{\lakl t} ^{\ka,\la} \embed =&   \e^{\i \lakl t M } \,  \mathop{\sum}\limits_{n \in  2\bbN}    \, \,  \mathop{\int}\limits_{0 \leq u_1 \leq  \ldots  \leq u_n \leq t } \d u_1 \ldots  \d u_n     \mathop{\sum}\limits_{\pi \in \Pair(2n)}   &         \label{def: dysonpaired1} \\ [0.7ex]
 &    \la^{-n}     \mathop{\int}\limits_{\caZ^n } \d z_1 \ldots \d z_n     \, \left(\mathop{\prod}\limits_{ (i,i') \rightarrow \pi} C(z_i,z_{i'}) \right)   \,         Q^{\ka}_{\lakl u_n,z_n}  \ldots   Q^{\ka}_{\lakl u_1,z_1}  &   \label{def: dysonpaired2}
\end{eqnarray}
It follows  from the definition \eqref{def: Z}, the Dyson expansion \eqref{def: series kappa} and the Wick theorem. \\

 Let  $[0,t]_{2}$ be  the set of (unordered) couples in  $[0,t]_{}$ and
\beq
\dsig_t:= \left\{ \dsi \subset [0,t]_{2}, \str\dsi \str < \infty \right\} 
\eeq
We remark that there is an idenfification between  $n \in 2 \bbN , 0 \leq u_1 \leq \ldots \leq u_n \leq t, \pi \in \mathrm{Pair}(n)$  and $\dsi \in \dsig_t$ with $\str\dsi\str=n/2$, given by 
\beq
\dsi =\mathop{\cup}\limits_{ (i,i') \rightarrow \pi}  \{ ( u_i,u_{i'}  )  \}
\eeq

   By writing $\d n$ and $\d_n \pi$ for the counting measures on respectively $\bbN$  and $ \mathrm{Pair}(n)$, we define, using the above  idenfification,  
\beq 
\d \dsi := \d n \times \d u_1 \times  \ldots \times \d u_n \times  \d_n \pi,   
\eeq
This definition could be ambiguous when $\str \dsi \str =0$ (hence $\dsi=\emptyset$), which we fix by defining \beq \int_{\dsig_{t} }  \d \dsi \, \mathrm{Ind}(\dsi=\emptyset)=1.\eeq
Thus, we have made $\dsig_t$ into a measure space.  Using the same identification, we define $\caV^{\ka,\la}(\dsi) \in \caB(\caB(\caE))$ to equal the line  \eqref{def: dysonpaired2} \beq \label{def: caV}
\caV^{\ka,\la}(\dsi):=  \la^{-n}     \mathop{\int}\limits_{\caZ^n } \d z_1 \ldots \d z_n     \, \left(\mathop{\prod}\limits_{ (i,i') \rightarrow \pi} C(z_i,z_{i'}) \right)   \,         Q^{\ka}_{\lakl u_n,z_n}  \ldots   Q^{\ka}_{\lakl u_1,z_1} 
\eeq
and we again abbreviate $\caV(\dsi):=\caV^{\ka,\la}(\dsi)$.

We have hence rewritten (\ref{def: dysonpaired1}-\ref{def: dysonpaired2})  as 
\beq  \label{eq: dysondsi}
\compress Z_{\lakl t} ^{\ka,\la} \embed =  \e^{\i \lakl t M } \,   \int_{\dsig_t} \d \dsi    \caV(\dsi) 
\eeq
For convenience, we also define $\basig_t \subset \dsig_t$ as the set of those $\dsi$ with $\str \dsi \str=1$. Hence $\basig_t$ is the set of ordered pairs in $[0,t]$. We will write the elements fo this pair as $\underline{s}(\basi),\overline{s}(\basi)$ with $\underline{s}(\basi) <\overline{s}(\basi)$. 

We stress that up to this point, nothing happened; we just cooked up a fancy notation, culminating in equation \eqref{eq: dysondsi},  for the Wick-ordered Dyson expansion!

\subsection{Proof of Lemma \ref{lem: bound caW}. }
 
 Statement (1) of Lemma \ref{lem: bound caW} is an obvious consequence of the definition \eqref{def: caW},  we concentrate on Statement (2).
We first establish the  crude a-priori bound \eqref{apriori bound}.

Let $(u_{a})$ be a basis in $\caE$ and define 

\beq \label{def: modified correlation} q_\ka(t) := \sum_{a,a',a'',a''''} \left\str \int_{\caX} \d x  \langle u_a, (V(x))^* u_{a'} \rangle    \langle u_{a''}, V(x) u_{a'''} \rangle \,  \e^{\i t \xi(x) }   w_{\ka}(x) \right\str     \eeq

Since $\caE$ is finite-dimensional, the function $ q_\ka(t) $  is dominated by a multiple of $ p_\ka(t) $ (as defined in \ref{def: correlation}) and vice versa. 
Using the explicit expression  \eqref{def: explicitQ}, \eqref{def: pairing coeff} and \eqref{def: caV}, one gets
 \beq
 \norm \caV(\basi) \norm \leq       \lakl
     ( q_{\Re \ka} +q_{\Im \ka})  (\frac{\overline{s}(\basi)-\underline{s}(\basi)}{\la^2} )+  \lakl( q_{\Re \ka}+q_{-\Im \ka})  (-\frac{\overline{s}(\basi)-\underline{s}(\basi)}{\la^2} ) =:  \lakl d_{\ka} (\frac{\overline{s}(\basi)-\underline{s}(\basi)}{\la^2} )  \eeq

One easily checks
 \beq
 \norm \caV(\dsi) \norm \leq         \mathop{\prod}\limits_{\basig_{t } \ni \basi \subset \dsi}  \lakl  d_{\ka} (\frac{\overline{s}(\basi)-\underline{s}(\basi)}{\la^2} ) =: G(\dsi) \eeq
(For example, one can represent $V^{{\#}}_u(x)= \sum_{a,a'}  \str a\rangle \langle a, V^{{\#}}_u(x) a' \rangle \langle a' \str $ in \eqref{def: explicitQ} and then factorize \eqref{def: caV}).
By a change of integration variables, and summing over all values of $\str \dsi \str$, we arrive at the a-priori bound
\beq\label{apriori bound}
\int_{ \dsig_t} \d \dsi   \norm \caV (\dsi) \norm   \leq     \e^{ t   \norm d_\ka \norm_1  }
\eeq
 with 
 $\norm d_\ka \norm_1= \int_{\bbR^+}d_\ka(t)\d t$, which is finite since $\norm p_\ka \norm_1$ is finite.

Let 
\beq J_{s,\tau_{\footnotesize{\pm }  }}(\dsi) := \mathrm{Ind} [ \exists \basi \in \basig_t , \, \basi \subset \dsi,\,  \underline{s}(\basi)\leq s \leq \overline{s}(\basi), \,  \overline{s}(\basi)- \underline{s}(\basi) \gtrless  \tau    ]  \eeq
%\begin{eqnarray} J_{s,\tau_{\footnotesize{\pm }  }}(\dsi) &:=&  \mathrm{Ind} [ \exists \basi \in \basig_t , \, \basi \subset \dsi,\,  \underline{s}(\basi)\leq s \leq \overline{s}(\basi), \,  \overline{s}(\basi)- \underline{s}(\basi) \gtrless  \tau    ] \\  
%J_{s,\tau}(\dsi) &:=&    \max \big\{ J_{s,\tau_{+}}(\dsi) , J_{s,\tau_{-}}(\dsi)   \big\}\end{eqnarray} 
One can easily convince oneself that ($\vee$ stands for the maximum)
\beq \label{relation WV}
\caW_n = \e^{\i n \lakl \tau M} \int_{\dsig_{n \tau}}  \d \dsi  \left(  \prod_{j =1}^{n-1} J_{j \tau,\tau_{-}}  \vee J_{j \tau,\tau_{+}}  \right) (\dsi) \caV(\dsi)
\eeq
In words, each $\dsi$ contributing to $\caW_n$ contains for each $j =1,\ldots,n-1$  a $\basi$ which 'crosses' $j \tau$.  Or, the insertion of $1-\embed \compress$ forces a pairing to occur.

\begin{lemma}\label{lem: estimates on caV} Assume Assumption \ref{ass: correlation}.
There are $c_{\pm}:=c_{\pm}(\ka,\la, \tau)$ vanishing as $\la \downarrow 0$ and continuous in the three parameters, such that for
 $\caJ_-, \caJ_+$ disjoint subsets of $\bbN_0$,
 \beq
  \mathop{\int}\limits_{\dsig_t} \d \dsi  \left( \mathop{\prod}\limits_{j_{\pm} \in \caJ_{\pm}} J_{j_{\pm} \tau,\tau_{\pm}} \right)   (\dsi)    \norm \caV(\dsi) \norm    \leq  (c_{+})^{\str \caJ_+ \str} (c_{-})^{\str \caJ_- \str}  \int_{\dsig_t} \d \dsi \,  G(\dsi)    \label{eq: estimate on caV3}
  \eeq
%\begin{eqnarray}
% & \left  \norm  \mathop{\int}\limits_{\dsig_t} \d \dsi    \caV(\dsi) \mathop{\prod}\limits_{j_- \in \caJ_-, \,  j_+ \in \caJ_+} J_{j \tau,\tau_{-}} (\dsi)  J_{j \tau,\tau_{+}}  (\dsi)  \right \norm &  \nonumber \\ \leq&  (c_{+})^{\str \caJ_+ \str} (c_{-})^{\str \caJ_- \str}  \int_{\dsig_t} \d \dsi \norm \caV(\dsi) \norm &    \label{eq: estimate on caV3}
%  \end{eqnarray}
\end{lemma} 
\begin{proof}

 Denote by $\{ \caG_i\}_i$ a partition of $\caJ_+$  in subsets $\caG_i$ satisfying $\max \caG_i < \min \caG_{i+1}$ for all $i$.  Write $\sum_{ \{\caG_{i} \}_{i} }$ for the sum over all such partitions.  Then, 
 \begin{eqnarray} 
\textrm{LHS of \eqref{eq: estimate on caV3}} &\leq&  \left( \sum_{ \{\caG_{i} \}_{i} } \prod_{i}  \int_{\basig_t} \d \basi \left( \prod_{j \in \caG_i}   J_{j \tau,\tau_{+}} (\basi) \right) \, G(\basi)  \right)  \label{proof: caV1} \\
 &\times&
  \left( \prod_{ j_{- } \in \caJ_{-}} \int_{\basig_t} \d \basi J_{j \tau,\tau_{-}}(\basi) G(\basi)  \right)  \,  \times \,  \left( \int_{\dsig_t} \d \dsi \, G(\dsi)  \right) \label{proof: caV2}
\end{eqnarray}
 %
 %
% If $ J_{j \tau,\tau_{+}}(\dsi) =1$ for $j \in \caJ_+$, there is a partition of $\caJ_+$ in subsets $\caG_i$ satisfying $\max \caG_i < \min \caG_{i+1}$ and such that for all $i$, there is a $\basi_i \subset \dsi$ such that 
% \beq
% \prod_{j \in \caG_i}   J_{j \tau,\tau_{+}}(\basi_i)=1
% \eeq
%
%
Since
\begin{eqnarray}
\int_{\basig_t} \d \basi J_{j \tau,\tau_{-}}(\basi) G(\basi)  &\leq&  \lakl  \int_{\underline{s} \leq s\leq \overline{s}} \d \overline{s} \d\underline{s}  \,d_{\ka}(\lakl (\overline{s}-\underline{s} )) \\ 
& \leq &   \la^2 \int_{\bbR^+}   \d u \,  u \, d_{\ka}(u)   =:c_{-},
  \end{eqnarray}
 hence the first factor in \eqref{proof: caV2} is bounded by $(c_{-})^{\str  \caJ_- \str }$.

By the argument following \eqref{def: modified correlation}, Assumption  \ref{ass: correlation} implies that there are $C_\ka,\al_\ka >0$ such that $d_\ka(t) <C_\ka \e^{-\al_\ka  \str t \str }$ (Obviously, $C_\ka;\al_\ka$ can  be chosen constant if $\ka$ varies in a bounded set).
One can bound
\beq
 \int_{\basig_t}  \d \basi \left(  \prod_{j \in \caG_i}   J_{j \tau,\tau_{+}}(\basi) \right)G(\basi) \leq    \frac{C_\ka}{\al_\ka}  \e^{- \al_\ka \lakl  (1/2) \str \max\caG_{i}-\min \caG_{i} \str \tau   }.
\eeq
%\beq
% \prod_i \int \d \basi_i  \prod_{j \in \caG_i}   J_{j \tau,\tau_{+}}(\basi_i) \norm \caV(\basi_i)\norm \leq      \frac{\la^2}{\al_\ka}  \e^{- \al_\ka  (1/2) \str \max\caG-\min \caG \str   }
%\eeq
Using  $\sum_{i} \str \max\caG_{i}-\min \caG_{i} \str \geq \str \caJ_{+}\str$ , we arrive at  the upper  bound for the RHS of \eqref{proof: caV1}
\beq
\sum_{ \{\caG_{i} \}_{i} } \prod_{i} \int_{\basig_t}  \d \basi   \left( \prod_{j \in \caG_i}  J_{j \tau,\tau_{+}}(\basi) \right)G(\basi)\leq   \e^{\frac{C_\ka}{\al_\ka}  \str \caJ_{+ }  \str  \tau }   \e^{- \al_\ka \lakl  (1/2)  \str \caJ_{+}\str \tau   } =: (c_{+})^{\str \caJ_{+} \str  }
\eeq

\end{proof} 

To conclude the proof of Lemma \ref{lem: bound caW}, we use expression \eqref{relation WV}, replacing $\vee \rightarrow +$,
\beq
\norm \caW_{n }  \norm  \leq   \int_{\dsig_t} \d \dsi    \left(  \prod_{j =1}^{n-1}  (J_{j \tau ,\tau_{-}}+  J_{j \tau ,\tau_{+}})   \right) (\dsi)  \norm \caV(\dsi)\norm   \leq   (c_{-}+c_{+})^{n-1}   (\e^{  \tau \norm d_\ka \norm_1})^n\label{lastproofline2} 
\eeq
%\begin{eqnarray}
%\norm \caW_{n }  \norm  &\leq&   \int \d \dsi      \left(  \prod_{j =1}^{m-1} J_{j \tau,\tau} \right) (\dsi) \caV(\dsi)        \\
% &\leq&   \int \d \dsi    \left(  \prod_{j =1}^{m-1}  (J_{j \tau ,\tau_{-}}+  J_{j \tau ,\tau_{+}})   \right) (\dsi) \caV(\dsi)  \label{lastproofline2}  \leq   (c_{-}+c_{+})^m  
%%  &\leq&      \left(      \left(  ( 1+ \frac{\la^2}{\al_\ka} ) \e^{- \al_\ka \tau}\right)^{\frac{1}{2  } }+   \la^2 \int \d u  uh(u)   \right)^m   \label{lastproofline3} 
%\end{eqnarray}
To get the last inequality, we represented the product in $\prod_{j =1}^{n-1}  (J_{j \tau ,\tau_{-}}+  J_{j \tau ,\tau_{+}}) $ as a sum over partitions of $\{ 1,\ldots,n-1\}$ in $2$ sets $\caJ_{-} $ and  $\caJ_{+} $, we applied Lemma \ref{lem: estimates on caV}  and we resummed the sum over partitions by the binomial formula. Finally, the bound \eqref{apriori bound} with $t=n \tau$ was used.

\subsection{Proof of Theorem \ref{thm:  properties L}}

These statements are contained in the literature. Statement (1) is a consequence of the Perron-Frobenius theorem for completely positive maps, stated in \cite{evanshoegh}  and valid in our context under Assumption \ref{ass: fermi golden rule} (This is extensively  discussed in \cite{deroeckmaesfluct}). 
 Statement (2) can be immediately checked from the explicit expressions in Section \ref{def: lind}. 
For $\ka=0$, Statement (3) is a result of the usual weak-coupling theory, see e.g.\ \cite{derezinskideroeck2}. For $\ka \neq 0$, it is a straightforward generalization of these theorems. One can easily follow the arguments in \cite{derezinskideroeck2} and adapt the statements.

%An even easier proof goes as follows:
%\ben
%\item{For $\Re \ka=0$, the statement actually follows immediately from \cite{derezinskideroeck2} since $\Ga(w_\ka)$ is then a bounded function of the operators $h_k$. }
%\item{  Under Assumption \ref{ass: laplace}, the expression    is analytic in $\ka$  }
%\een

\section*{Acknowledgments}
The author has benefited from good discussions with H.-T. Yau, C.
Maes, J. Bricmont, J. Derezi\'{n}ski, J. Fr\"{o}hlich and C-A. Pillet.  Constructive critique on a previous version of the manuscript was raised by W. Abou-Salem, G.M. Graf and D. Spehner. The
financial support of the FWO-Flanders is greatly acknowledged.

\bibliographystyle{plain}
\bibliography{mylibrary07specialforness}

\end{document}